\documentclass[letterpaper, paper,11pt]{AAS}		% for final proceedings (20-page limit)

\usepackage{amsmath,amsfonts,amsthm}
\usepackage{bm}
\usepackage{color}
\usepackage{comment}
\usepackage{enumitem}
\usepackage{float}
\usepackage{footnpag}% make footnote symbols restart on each page
\usepackage[colorlinks=true, pdfstartview=FitV, linkcolor=black, citecolor= black, urlcolor= black]{hyperref}
\usepackage{cleveref}
\usepackage{parskip} %sets parindent to zero
\usepackage{overcite}
\usepackage{subfigure}

\newcommand{\eqn}[1]{\begin{equation} #1 \end{equation}}
\newcommand{\eqns}[1]{\begin{equation*} #1 \end{equation*}}
\let\oldalpha\alpha
\renewcommand{\alpha}{\bm{\oldalpha}}
\let\oldbeta\beta
\renewcommand{\beta}{\bm{\oldbeta}}
\let\oldgamma\gamma
\renewcommand{\gamma}{\bm{\oldgamma}}
\let\oldmu\mu
\renewcommand{\mu}{\bm{\oldmu}}
\let\oldOmega\Omega
\renewcommand{\Omega}{\bm{\oldOmega}}
\let\oldphi\phi
\renewcommand{\phi}{\bm{\oldphi}}
\let\oldSigma\Sigma
\renewcommand{\Sigma}{\bm{\oldSigma}}
\let\oldzeta\xi
\renewcommand{\zeta}{\bm{\oldzeta}}

\renewcommand{\b}{\bm{b}}
\newcommand{\B}{\bm{B}}
\newcommand{\f}{\bm{f}}
\newcommand{\F}{\bm{F}}
\newcommand{\g}{\bm{g}}
\newcommand{\G}{\bm{G}}

\newcommand{\HU}{\bm{H}_U}
\newcommand{\Hy}{\bm{H}_{\bm{y}}}
\newcommand{\J}{\bm{J}}
\newcommand{\lm}{\lambda_m}
\newcommand{\lM}{\lambda_M}
\newcommand{\phiF}{\phi_{\S}}
\newcommand{\phix}{\phi_{\X}}
\newcommand{\omegab}{\bm{\omega}}
\newcommand{\Q}{\bm{Q}}
\renewcommand{\r}{\bm{r}}
\newcommand{\R}{\mathbb{R}}
\newcommand{\s}{\bm{s}}
\renewcommand{\S}{\bm{S}}
\newcommand{\Us}{\bm{U}_{\bm{s}}}
\newcommand{\x}{\bm{x}}
\newcommand{\n}{\bm{n}}
\renewcommand{\v}{\bm{v}}
\newcommand{\X}{\bm{X}}
\newcommand{\y}{\bm{y}}
\newcommand{\yx}{\bm{y}_{\z}}
\newcommand{\z}{\bm{z}}

\DeclareMathOperator{\tr}{\mathrm{tr}}
\DeclareMathOperator{\diag}{\mathrm{diag}}
\DeclareMathOperator{\cov}{\mathrm{cov}}
\DeclareMathOperator{\corr}{\mathrm{corr}}
\DeclareMathOperator{\id}{\mathbf{I}}

\newtheorem{thm}{Theorem}
\newtheorem{cor}{Corollary}

\PaperNumber{19-368}

\begin{document}

%\title{UNCERTAINTY QUANTIFICATION AND ANALYSIS OF INVARIANTS OF RIGID BODY ATTITUDE DYNAMICS}
\title{UNCERTAINTY QUANTIFICATION AND ANALYSIS OF DYNAMICAL SYSTEMS WITH INVARIANTS}
\author{Anant A. Joshi\thanks{Undergraduate Student, Department of Mechanical Engineering, Indian Institute of Technology Bombay, India.}  
\ and Kamesh Subbarao\thanks{Associate Professor, Department of Mechanical and Aerospace Engineering, The University of Texas at Arlington, USA.}
}
\maketitle{}

\begin{abstract}
This paper considers uncertainty quantification in systems perturbed by stochastic disturbances, in particular, Gaussian white noise. The main focus of this work is on describing the time evolution of statistical moments of certain invariants (for instance total energy and magnitude of angular momentum) for such systems. A first case study for the attitude dynamics of a rigid body is presented where it is shown that these techniques offer a closed form representation of the evolution of the first and second moments of the kinetic energy of the resulting stochastic dynamical system. A second case study of a two body problem is presented in which bounds on the first and second moments of the angular momentum are presented.
\end{abstract}

\begin{comment}
\section{Literature Survey}
Methods of solutions to Fokker-Planck equation \cite{risken}
Stationary solutions to Fokker-Planck equation for Hamiltonian systems \cite{fuller-fpham}. 
Solution of Fokker-Planck equation for deterministic Hamiltonian system i.e. with no perturbations modelled as random variables) and specialized to the case that the initial distribution is Gaussian \cite{park-nlg}. They use Taylor series expansion and give a recipe to decide where to terminate the expansion. 
Another interesting approach to uncertainty quantification is to propagate statistical moments only, instead of propagating an approximation to the entire state \cite{singla-qgkf}. 
\end{comment}

\section{Introduction}
%Uncertainty quantification and study of invariants have been interesting topics.

Uncertainty quantification deals with analysis of how the probability distribution of the states of a system changes with time. Various fields of sciences have employed uncertainty quantification techniques \cite{fp-astro,fp-physics,fp-chem1,fp-chem2,fp-chem3} in addition to engineering fields astrodynamics, dynamical systems and estimation. The Fokker-Planck equation\cite{risken,junkins-optest} allows us to quantify this temporal variation; however being a partial differential equation in both time and states, it is non-trivial to solve. Attempts at simplifying the partial differential equation by considering stationary conditions have been madea as in Ref.~[\citenum{fuller-fpham}].   A commonly employed approach to quantify uncertainty is to approximate the probability density function (pdf) of the state of a dynamical system as a sum of Gaussians. In a typical estimation problem, the square of the difference between the estimate of the pdf and the true predicted pdf \cite{singla-gmm} is minimized as in Gaussian sum filtering using weights that change over time (adaptive weights) \cite{singla-gsf} as opposed to keeping them constant \cite{sorenson-1,sorenson-2} between successive measurements. Further, the number of components in the Gaussian mixture too can be varied over time \cite{singla-split}. Gaussian approximation can also be propagated in time using Taylor series expansion of the dynamical equation governing the system \cite{park-nlg}. Additionally, uncertainty quantification for systems that evolve on manifolds has also been an area of interest \cite{lee-1,kdm-1,kdm-2} and has found uses in estimation\cite{lee-2,sanyal-1}. 

In this paper we focus on certain invariants of dynamical systems defined as quantities which do not change over time. For instance the total mechanical energy of an ideal autonomous spring mass system is an invariant. Studying invariants is useful as they give us information about the bounds on the states of a system, and being real valued quantities they are easy to manipulate. Controller design using Lyapunov method uses exactly this principle - to introduce such a control that will drive a positive definite invariant for the original system to zero. For the two body orbit dynamics problem, the invariants total energy and angular momentum, gives us details about the orbit of the orbiting body. For a torque free rigid body, kinetic energy and magnitude of angular momentum give us details about the stability of the body when analysed using polhode plots\cite{wiesel-polhode,hughes-polhode}. Given such invariants in a dynamical system, we investigate the effect of random noise on the dynamical systems and thereby the invariants. 

For Hamiltonian systems without random perturbations, the pdf at any time can be computed if it is known at some prior time instant, and this approach can also be used to approximate the state as a Gaussian random variable \cite{park-nlg}. For Hamiltonian systems with random perturbations, solution to the Fokker Planck equation exists \cite{fuller-fpham}, but under stationary conditions. In this work we will quantify the temporal evolution of first and second statistical moments of scalar invariants for dynamical systems perturbed by Gaussian white noise. We will make use of the statistical properties of Gaussian white noise. We will begin with a brief overview of probability and Brownian motion followed by a rigorous description of invariants and the problem description. Then we will present the main theorem in this paper, and finally its applications in finding the time evolution of the first two statistical moments of the states and invariants. We will also present application of the developed theory to two case studies: first we study rigid body dynamics in which we investigate the evolution of kinetic energy of the rigid body and present interesting and useful results which are verified by numerical simulation; second we study the two body problem in which we investigate the evolution of angular momentum and provide bounds on its rate of change.

\section{Mathematical Preliminaries}

\subsection{Probability Overview}
The mean of a random variable $\X$ will be denoted by $E[\X]$. If the mean is a function of time (if for example, $\X$ is a stochastic process), mean at time $t$ will be denoted by $E[\X](t):=E[\X(t)]$ and covariance at time $t$ will be denoted as $\cov[\X](t):=E[(\X(t) - E[\X](t))(\X(t) - E[\X](t))^T]$. The correlation of $\X$ will be denoted as $\corr[\X](t) = E[\X(t)\X(t)^T]$. The covariance is related to the correlation as 
\eqn{
\cov[\X](t) = \corr[\X] - E[\X]E[\X]^T
\label{eq:covcorr}
} 

Differentiating this we have
\eqn{
\frac{d}{dt}\cov[\X] = \frac{d}{dt}\corr[\X] - (\frac{d}{dt}E[\X])E[\X]^T - E[\X](\frac{d}{dt}E[\X])^T
\label{eq:covcorrder}
}
The probability density function of $\X$ will be denoted by $p_{\X}(\cdot)$. 
The mean of a function of a random variable $\y(\X)$ is given by$ E[\y(\X)] = \int \y(\alpha)p_{\X}(\alpha)d\alpha $. 
For a multivariate Gaussian random variable $\bm{G} \in \R^g$, with probability density function $p_{\bm{G}}(\cdot)$, mean $\bm{\mu_G} \in \mathbb{R}^g$ covariance $\bm{Q_G} \in \R^{g\times g}$, the mean and covariance of a linear transformation $\bm{LG}$ of $\bm{G}$ is 

\begin{equation}
\label{eq:grvcov}
E[\bm{LG}] = \bm{L\mu_G}, \qquad \cov[\bm{LG}] = \bm{LQ_GL}^T 
\end{equation}
The moment generating function of the multivariate Gaussian is defined as $\varphi(\bm{t}) := \exp (\iota \bm{t}^T\bm{\mu_G} - \frac{1}{2}\bm{t}^T\bm{Q_G}\bm{t})$ for $\bm{t} \in \R^g$ and with $\iota:= \sqrt{-1}$. The following results will be used later. The proofs are trivial.
For a multivariate Gaussian random variable $\bm{G} \in \R^g$, $E[\bm{G}^T\bm{MG}] =  \tr(\bm{M}\corr[\bm{G}])$ for $\bm{M} \in \R^{g\times g}$.
Given two random variables $\X_1$ and $\X_2$ with probability density functions $p_{\X_1}(\cdot)$ and $p_{\X_2}(\cdot)$ respectively, the mean of $\X_1$ can be written as $E[\X_1] = \int \left(\int \alpha p_{\X_2}(\beta)d\beta\right)p_{X_1}(\alpha)d\alpha$.

%%%%
\subsection{Brownian Motion Overview}
Brownian motion ($\B(t)\in \R^m, t \in \R$) is a stochastic process having the following properties (Refs.~[\citenum{evans-bm, junkins-optest}]):
\begin{enumerate}
	\item $\B(0) = 0$ almost surely
    \item {The Browninan motion has independent increments}
    \item $\B(t) - \B(s) \sim \mathcal{N}(0,(t-s)\Q)$
    \item $d\B(t) := \B(t+dt) - \B(t) \implies d\B \sim \mathcal{N}(0,\Q dt)$
\end{enumerate}
where $dt \in \R$ denotes an infinitesimal time increment here and throughout the paper.
{
We will further assume that 
\begin{enumerate}
\setcounter{enumi}{4}
\item the Brownian motion is independent of the state
\item $\Q$ is diagonal.
\end{enumerate}

Using the results in Ref.~[\citenum{GRV}] it can be observed that 
\begin{enumerate}
	\item ${E[d\B_i(t)d\B_j(t)d\B_k(t)] = 0} \qquad \forall \quad i, j, k = 1,2,3,\ldots,m$ 
    \item ${E[d\B_i(t)d\B_j(t)d\B_k(t)d\B_l(t)] = O((dt)^2)} \qquad \forall \quad i, j, k, l = 1, 2, 3,\ldots,m$ 
\end{enumerate}
}

\subsection{Gaussian white noise}
Heuristically, the time derivative of Brownian motion is referred to as white noise $\n:=\frac{d\B}{dt}$ although Brownian motion is not differentiable. When writing a stochastic differential equation, Brownian motion is often represented as white noise (See Ref.~[\citenum{evans-bm}] for more details). In this work, systems perturbed by Brownian motion will be said to be systems perturbed by Gaussian white noise interchangeably. Note however, that both point to the same object, but they are just different words used in different cases. For example, \cref{BMsys,WNsys} will represent the same physical system but have different mathematical precision of expression. 
\eqn{ \label{BMsys}
d\x(t) = d\B(t)
}
\eqn{ \label{WNsys}
\dot{\x}(t) = \n(t)
}

\subsection{Flow of a deterministic system}
%\label{sec:flowprop}
Assume that the state $\x \in \R^n$ of a dynamical system is governed by the following deterministic differential equation:

\begin{equation} \label{detsys}
	\dot{\x}(t) = \f(t,\x) 
\end{equation}

$\x(t) \in \mathbb{R}^n , \f: \mathbb{R} \times \R^n \rightarrow \mathbb{R}^n$ with initial condition $\x(t_0) = \x_0$ \par

The solution flow of the system is a map $\phi_{d,\X}:\R \times \R \times \R^n \rightarrow \R^n $ such that

\begin{enumerate}
\item $\phi_{d,\X}(t_0,t_0,\x_0) = \x_0$
\item $\frac{\partial}{\partial t} \phi_{d,\X}(t,t_0,\x) = \f(t,\phi_{d,\X}(t,t_0,\x))$
\item $\phi_{d,\X}(t,t_0,\x_0) = \x(t)$
\end{enumerate}

\subsection{Description of stochastic systems}
Assume that the system in \cref{detsys} is perturbed by Gaussian white noise. The state of the resulting stochastic system is governed by the following stochastic differential equation:
\begin{equation} \label{stosys}
	d\x = \f(t,\x)dt + \g(t,\x)d\B(t) 
\end{equation}
$\x(t) \in \mathbb{R}^n , d\B(t) \in \mathbb{R}^m , \f:\mathbb{R} \times \mathbb{R}^n \rightarrow \mathbb{R}^n , \g: \mathbb{R} \times \mathbb{R}^m \rightarrow \mathbb{R}^{n \times m}$. The state at time $t$ will be a random variable, denoted by $\X(t)$.

To describe the flow of a stochastic system, we will define a function $\phi$, inspired by the first order Taylor series expansion of the flow of a deterministic system, suitably modified for our use. If $\g\equiv0$, \cref{stosys} is a deterministic system and the flow $\phi_{d,\X}$ satisfies properties of the flow of a deterministic system. Hence,
$\phi_{d,\X}(t_0+dt,t_0,\x_0)= \x_0 + \f(t_0,x_0)dt + O((dt)^2)$

Define $\phi_{\X} : \R \times \R \times \R ^n \times \R ^m \rightarrow \R^n$ to satisfy 

\begin{align}
	\phi_{\X}(t_0 + dt,t_0,\alpha,\b) =& \alpha  + \f(t_0,\alpha)dt + \g(t_0,\alpha)\b + O((dt)^2) \label{flow2} \\
	\phi_{\X}(t_0,t_0,\x,\b) =& \x \label{flow2-2}
\end{align}

with $\b$ defined for readibility as $\b := d\B(t)$. Note that $\b \sim N(0,\Q(t)dt)$ is a random vector. It denotes realization of the Brownian motion.

In the rest of the document, the flow will only be expanded to first order in $dt$ by dropping the $O((dt)^2)$ terms since they will vanish on taking $\lim_{dt \to 0}$. For more literature on stochastic flows, the reader is referred to Refs.~[\citenum{sflow1, sflow2}].

$p_{\X}(t,\alpha)$ will refer to the probability density of the state $\X$ which is also a function of time $t$. In other words, $p_{\X}(t,\alpha)d\alpha$ is the probability that $\X \in (\alpha,\alpha + d\alpha)$ at time $t$. $p_{\b}(\cdot)$ applies similarly to $\b$.

\section{Problem Description}

\noindent \textit{Invariant for deterministic system:}
For a deterministic system invariants are quantities which do not change along the flow of the system. In the present context, the invariants are functions of the state of the system. For example, for a spring-mass system, the invariant (total energy) is a function of the state vector (position, and velocity). We shall use the notation $\S$ to capture all those variables of which the invariant is a function of. 
\eqn{
\dot{\s}(t) = \F(t,\s) 
\label{detsysinv}
} 
$\s(t) \in \mathbb{R}^q , \F: \mathbb{R} \times  \mathbb{R}^q \rightarrow \mathbb{R}^q$ will denote the dynamics of the variables comprising $\S$ for the system in \cref{detsys} and $\phi_{d,\S}$ the flow of the flow of system in \cref{detsysinv}. We shall use the notation $U:\R^q\rightarrow \R $ to denote scalar invariants for a system. For an invariant $U(\s)$ of the system in \cref{detsys} the dynamics of the invariant yields

\eqn{ \label{eq:invprop}
	\frac{dU}{dt}=0\implies (\Us(\alpha))^T \F(\alpha) = 0 
}

for any state vector $\alpha$, where $\Us(\zeta) := \dfrac{dU}{d\s}\bigg\rvert _ {\s = \zeta}$

If the system in \cref{detsys} changes to \cref{stosys} then the dynamics governing the invariant will also change, which we will denote by 

\eqn{
d\s = \F(t,\s)dt + \G(t,\s)d\B(t) 
\label{stosysinv}
}

The state governing the invariant at time $t$ will be a random variable denoted by $\S(t)$.
{
Similar to \cref{flow2,flow2-2}, $\phiF(t+dt,t,\zeta,\b)$ will denote the flow of \cref{stosysinv}. $p_{\S}(t,\cdot)$ will denote the probability density of $\S$.
}{We address the problem of quantifying the first and second statistical moments of $U(\S(t))$}.

\section{Solution Approach}

\begin{thm} 
\label{th:T1}
Given a deterministic function $\y$ of the state $\X$, define $\displaystyle \yx(\alpha) := \dfrac{dy}{d\z}\bigg\rvert _ {\z = \alpha}$ and $\displaystyle \Hy(\alpha) := \frac{1}{2} {\frac{d^2 \y}{d \z d \z^T }}\bigg\rvert _ {\z = \alpha}$ 

For the system in \cref{stosys},
\begin{eqnarray}
\dfrac{d}{dt}\left({E}[\y(\X)]\right)  & = & \int \yx(\alpha)^T\f(t,\alpha) p_{\X}(t,\alpha) d\alpha \nonumber \\ 
& & + \lim_{dt \to 0} \frac{\int \int ( (\g(t,\alpha)\beta) ^T \Hy(\alpha) \g(t,\alpha)\beta p_{\b}(\beta)d\beta ) p_{\X}(t,\alpha)d\alpha}{dt} \label{eq:dEdt}
\end{eqnarray}
\end{thm}

%\begin{proof}
%The result is derived by calculating the derivative of the expectation of $\y(\X)$ following the definitions stated earlier in the paper.
%\end{proof}

\begin{proof}
Only for the proof of \cref{th:T1} the notation for probability density will be changed. $p_{\X_t}(\alpha) := p_{\X}(t,\alpha)$ will denote the probability density of the random variable $\X$ at the time $t$. 

We will go through the proof in two steps. First, we will find the expectation of the function of the state due to an infinitesimal increment $dt$ in time. Then, we will use the first principles definition of time derivative to complete the proof. 

In \cref{flow2}, given the value of state ($\alpha$) and value of $\b$ ($\beta$) at time $t$, the state at time $t+dt$ will be given by 
$$\gamma := \X(t+dt) = \alpha + \f(t_0,\alpha)dt + \g(t_0,\alpha)\beta$$ 

Hence if we know the value of $\X(t)$ and $\b$ we can find the value of $\X(t + dt)$ exactly. If we (in an intuitive sense) average over the values of all $\alpha$ and $\beta$ (since both of them are random variables) we will arrive at $E[\X(t + dt)]$. This is intuitively equivalent to finding all values of $\gamma$ itself and averaging over them, since all values of $\alpha$ and $\beta$ will invariably lead to all values of $\gamma$. 

We want to cast this intuition in a probabilistic way so that we can write the expectation in terms of integrals. Therefore, written in the form of probability density function (denoting the Dirac delta function by $\delta(\cdot)$) 
\eqn{
p_{\X_{t+dt}|\X_t,\b}(\gamma|\alpha,\beta) = \delta(\gamma - (\alpha + \f(t,\alpha)dt + \g(t,\alpha)\beta))
\label{pdelta}
}  

Now that we know the value of the next state given the previous state and the realisation of Brownian motion, or equivalently the conditional probability density, we will try to find the probability density function of $\gamma$ and integrate over all $\gamma$ to find the expectation. 

By the definition of conditional probability (dropping arguments of functions for readability whenever required),
\begin{align}
p_{{\X_{t+dt},\X_t,\b}} =& p_{\X_{t+dt}|\X_t,\b} \ p_{\X_t,\b} 
\end{align} 

Since the Brownian motion is assumed independent of the state,
\eqn{
p_{{\X_{t+dt},\X_t,\b}} = p_{\X_{t+dt}|\X_t,\b} \ p_{\X_t} \ p_{\b} 
\label{pXXb}
}

From the definition of marginal probability density and expectation value,
\eqn{
p_{\X_{t+dt}}(\gamma) = \int \int p_{{\X_{t+dt},\X_t,\b}} (\gamma,\alpha,\beta) d\alpha d\beta
\label{pmarg}
}
\eqn{
E[\X](t+dt) = E[\X_{t+dt}] = \int \int \int \gamma p_{\X_{t+dt},\X_t,\b} d\gamma d\alpha d\beta
}

Substituting \cref{pXXb,pdelta}, recalling the properties of $\delta(\cdot)$ and integrating in $\gamma$ we have, 
\eqn{E[\X](t+dt) = \int \int (\alpha + \f(t,\alpha)dt + \g(t,\alpha)\beta)p_{\X_t} p_{\b} d\beta d\alpha}

Using the definition of $\phix$ in \cref{flow2}
\eqn{E[\X](t+dt) = \int \int \phix(t+dt,t,\alpha,\beta)p_{\X_t} p_{\b} d\beta d\alpha} 

This conforms with the intuitive idea of averaging over all possible values of $\alpha$ and $\beta$. We want to do the same for any function $y(X)$. Recalling, 
\eqn{ \label{eq:Ey1}
E[\y(\X)](t) = \int \y(\alpha)p_{\X}(t,\alpha)d\alpha
}

and
\eqn{p_{\y(\X_{t+dt})|\X_t,\b}(\gamma|\alpha,\beta) = \delta(\gamma - \y(\alpha + \f(t,\alpha)dt + \g(t,\alpha)\beta))}

we similarly obtain
\eqn{ \label{eq:Ey2}
E[\y(\X)](t+dt) = \int \int \y(\phix(t+dt,t,\alpha,\beta))p_{\X_t} p_{\b}d\beta d\alpha
}

We will now proceed to find the derivative having found the value at $t + dt$. Using multivariate Taylor series expansion upto second order along the dynamics of \cref{stosys} and recalling that $\y(\phix(0,t,\alpha,\beta)) = \y(\alpha)$ we have,
\eqn{ \label{ysecondorder}
	\y(\phix(t+dt,t,\alpha,\beta) = \y(\alpha) + (\yx(\alpha))^T(d\x) + d\x ^T \Hy(\alpha) d\x
}

Subtracting \cref{eq:Ey1} from \cref{eq:Ey2} and dropping arguments of functions we have
\eqn{
	dE[\y(\X)] = \int \int( (\yx)^T(\f dt + \g\beta) + (\f dt + \g\beta)^T \Hy(\f dt + \g\beta) )p_{\b}d\beta p_{\X_t}d\alpha
}

On expanding the integrand, the following terms will integrate to zero:

\begin{enumerate}
\item $\yx ^T \g\beta$, $\f ^T \Hy \g\beta dt$ and $(\g\beta) ^T\Hy \f dt$ since $\b$ is a Gaussian random variable with zero mean 
\item $\f ^T \Hy \f (dt)^2$ since $O((dt)^2)$ terms are ignored 
\end{enumerate}

Finally taking $\lim_{dt \to 0}$ and recalling $\int p_{\b}(t,\beta)d\beta = 1$

\eqn{
	\dot{E}[\y(\X)] = \int \yx ^T \f p_{\X} d\alpha + \lim_{dt \to 0}\frac{\int \int (\g\beta)^T \Hy  \g\beta p_{\b}d\beta p_{\X}d\alpha}{dt}
}

This completes the proof of the theorem.

Note : the limit makes sense since the integral in the numerator is $O(dt)$ because the covariance of $\b$ is $\Q dt$.
\end{proof}

We will present two special cases of the application of \cref{th:T1} and the results are consistent with similar ones as in Ref.~[\citenum{singla-qgkf}].

\begin{cor}
\label[Corollary]{eq:meanprop}
Evolution of mean of the state: Setting $\y(\X):=\X$, $\yx = \id$ and $\Hy = \mathbf{0}$. Using \cref{th:T1} we have $\dfrac{d}{dt}\left({E}[\X]\right) = E[\f(t,\X)]$
\end{cor}
\begin{cor}
\label[Corollary]{eq:corrprop}
Evolution of correlation of the state: Setting $\y(\X) = \X\X^T$, we see $\yx(\alpha)\cdot\f(\alpha) = \alpha\f^T(\alpha) + \f(\alpha)\alpha^T$ and using \cref{th:T1} we have, \[\dfrac{d}{dt}\left({E}[\X \X^T]\right) = E[\X \f^T(\X)] + E[\f(\X)\X^T] + E[\g(\X)\Q \g^T(\X)]\]
\end{cor}

$\frac{d}{dt}\cov[\X] = \frac{d}{dt}\corr[\X] - (\frac{d}{dt}E[\X])E[\X]^T - E[\X](\frac{d}{dt}E[\X])^T$ which is obtained by differentiating the formula $\cov[\X] = \corr[X] - E[\X]E[\X]^T$.

\subsection{Evolution of invariant}

We will derive expressions for the rate of change of the first and second moments of an invariant $U(\s)$ with underlying dynamics given by \cref{stosysinv}. Define $\displaystyle \HU(\alpha) := \frac{1}{2} {\frac{d^2 U}{d \s d \s^T }}\bigg\rvert _ {\s = \alpha}$
{ Also define,
\begin{eqnarray}
u_1(t)       & := & E[U(\s)](t) = \int U(\zeta) p_{\S} (t,\zeta)d\zeta\label{eq:expu} \\
u_2(t)       & := & E[(U(\s))^2](t) \label{eq:expu2-1} \\
\cov[U](t) & := &  E[(U(\s) - u_1)^2](t) = (u_2(t)) - (u_1(t))^2 \label{eq:expu2-2}
\end{eqnarray}
}
The mean, correlation, and covariance of the invariant are given by eqs.~\eqref{eq:expu}, \eqref{eq:expu2-1} and \eqref{eq:expu2-2} respectively. Using \cref{th:T1} and \cref{eq:invprop} we obtain on substitution (suppressing arguments of functions to keep notation clean):

\begin{eqnarray}
\dot{u}_1(t) & = & \int \tr(\HU\G\Q\G^T) p_{\S}d\zeta \label{eq:meaninv} \\
\dot{u}_2(t) & = & \int [2U\tr(Q\G^T\HU\G) + \tr(\Us^T\G\Q\G^T\Us)]p_{\S} d\zeta \label{eq:corrinv} \\
\frac{d}{dt}(\cov[U](t)) & = & \int [2(U -u_1) \tr(Q\G^T\HU\G) + \tr(\Us^T\G\Q\G^T\Us)]p_{\S} d\zeta \label{covinv}
\end{eqnarray}

It must be kept in mind that in cases where $\G$ is the identity map,  will greatly simplify these expressions. 

\section{Analysis of Rigid Body Attitude Dynamics}
For a rigid body, the dynamics of the angular velocity of the body with respect to an inertial frame, expressed in body frame components ($\bm{\omega} \in \R^3$) is governed by \cref{rbd} where the moment of inertia is denoted by $\J= \J^T > 0 \in \R^{3 \times 3}$ and torque acting on rigid body by $\bm{\tau}\in \R^3$. If we carry out the analysis with body axes assumed to be aligned along the principal axes of inertia then $\J$ is diagonal, and \cref{rbd} simplifies to \cref{rbd1,rbd2,rbd3}.
\begin{gather}
	\J\dot{\bm{\omega}} = -\bm{\omega} \times \J\bm{\omega} + \bm{\tau} \label{rbd} \\
    \dot{\omega}_1 = \frac{J_2-J_3}{J_1}\omega_2\omega_3 + \frac{\tau_1}{J_1} =  c_1\omega_2\omega_3 + n_1 \label{rbd1}\\
    \dot{\omega}_2 = \frac{J_3-J_1}{J_2}\omega_1\omega_3 + \frac{\tau_2}{J_2} =  c_2\omega_1\omega_3 + n_2 \label{rbd2}\\
    \dot{\omega}_3 = \frac{J_1-J_2}{J_3}\omega_1\omega_2 + \frac{\tau_3}{J_3} = c_3\omega_1\omega_2 + n_3 \label{rbd3}
\end{gather}

Clearly, for the rigid body, we identify $\f(t,\alpha) = -\J^{-1}(\alpha \times \J \alpha) $ and $\g(t,\alpha)=\J^{-1}$. 

For this rigid body subject to torque free motion, there are two invariants, the kinetic energy and the norm of angular momentum.
Both the invariants are functions only of the angular velocity of the rigid body. We will look at evolution of first and second moments of the kinetic energy in the presence of stochastic torques ($\Q$ will denote the covariance of $\bm{\tau}$). Denote the mean, correlation and covariance of angular velocity by $\mu, \bm{R}$ and $\Sigma$ respectively. Assume the body axes to be aligned with the principal axes of moment of inertia. We use Corollary  \ref{eq:meanprop} and  \ref{eq:corrprop} to arrive at expressions of the expressions for $\dot{\mu}$ and $\dot{\sigma}$ (refer to Appendix A for details of derivation). 
The expressions for $\dot{\mu}$ expressed in scalar form are as follows (the expressions for $\dot{\sigma}$) are not presented in scalar form since they are not very readable):

\eqns{\dot{\oldmu_1} = c_1\oldSigma_{23} + c_1\oldmu_2\oldmu_3}
\eqns{\dot{\oldmu_2} = c_2\oldSigma_{31} + c_2\oldmu_3\oldmu_1}
\eqns{\dot{\oldmu_3} = c_3\oldSigma_{12} + c_3\oldmu_1\oldmu_2}

Alternately, the combined equations can be expressed in vector form as 
\eqn{
\dot{\mu} = -\J^{-1}\left( \mu \times \J \mu \right) + \J^{-1}{\bm{\sigma}}
\label{eq:rbmeanprop}
}
\eqn{
\dot{\Sigma} = \bm{A}\Sigma + \Sigma \bm{A}^T  + \J^{-1}\Q\J^{-1}
\label{eq:rbcovprop}
}

where $\displaystyle \bm{A} := \frac{\partial}{\partial \alpha}\f(t,\alpha)\bigg\rvert_{\alpha = \mu}$ and $\displaystyle \bm{\sigma} := \left[(J_2 - J_3)\oldSigma_{23}~\quad (J_3 - J_1)\oldSigma_{31}~\quad (J_1-J_2)\oldSigma_{12}\right]^T$.

\noindent It is to be noted that if $\Sigma$ is diagonal, then evolution of the mean of the angular velocity states for the stochastic attitude dynamics reduces to 

\eqns{\dot{\mu} = -\J^{-1}\left( \mu \times \J \mu \right)}

which has the same structure as the torque free rigid body motion. We now look at the two invariants previously mentioned.

\subsection{Rotational Kinetic Energy}

The kinetic energy $U_K$ of a rigid body is given by the expression 
\eqn{U_K(\omegab) = \frac{1}{2}\omegab^T\J\omegab \label{enrbd}}
This yields  
\eqn{(U_K)_{\omegab} = \J\omegab \label{enrbd1}}
and \eqn{H_{U_K} = \frac{\J}{2} \label{enrbd2}}

\noindent \underline{Evolution of the Mean:} 
\begin{equation}
	\oldmu_K(t) := E[U_K(\omegab)](t)
%\label{expen1}
\end{equation}
Using eq. \eqref{eq:meaninv} and \cref{enrbd2} 
\eqn{
\dot{\oldmu}_K(t) = \frac{1}{2}\tr(\J^{-1}\Q) 
\label{eq:rbenmean}
}

\noindent \underline{Evolution of Correlation:}
\eqn{R_K(t):=E[(U_K(\omegab))^2](t)}
Using eq. \eqref{eq:corrinv} and \cref{enrbd1,enrbd2}
\eqn{ \dot{R}_K(t) = \oldmu_K(t)\tr(\J\Q) + \tr((\Sigma + \mu\mu^T)\Q)}

\noindent \underline{Evolution of Covariance:}
\eqn{
\oldSigma_K = \cov[U_K(\omegab)](t) = R_K(t) - \oldmu_K^2(t)
}

Using eq. \eqref{covinv}, \eqref{enrbd1} and \eqref{enrbd2} we have
\begin{equation}
\label{eq:rbencov}
\dot{\oldSigma}_K = \frac{d}{dt}\cov[U_K(\omegab)] =  \tr((\Sigma + \mu\mu^T)\Q)
\end{equation}

Refer Appendix B for details of derivations. 
\newpage
\noindent \underline{Summary:}

Thus the governing equations for the first and second moments of the states and the corresponding invariant (rotational kinetic energy) can be summarized as follows:
\begin{eqnarray*}
\dot{\mu} & = & -\J^{-1}\left( \mu \times \J \mu \right) + \J^{-1}{\bm{\sigma}} \\
\dot{\Sigma} & = & \bm{A}\Sigma + \Sigma \bm{A}^T  + \J^{-1}\Q\J^{-1} \\
\dot{\oldmu}_K(t) & = & \frac{1}{2}\tr(\J^{-1}\Q) \\
\dot{\oldSigma}_K & = & \tr((\Sigma + \mu\mu^T)\Q)
\end{eqnarray*}
where $\displaystyle \bm{A} := \frac{\partial}{\partial \alpha}\f(t,\alpha)\bigg\rvert_{\alpha = \mu}$ and $\displaystyle \bm{\sigma} := \left[(J_2 - J_3)\oldSigma_{23}~\quad (J_3 - J_1)\oldSigma_{31}~\quad (J_1-J_2)\oldSigma_{12}\right]^T$.

\subsection{Numerical Verification}
We observe that it is difficult to find an analytical solution for the mean and covariance propagation of angular velocity (\cref{eq:rbmeanprop,eq:rbcovprop}) in the general case. However, the structure of \cref{eq:rbmeanprop,eq:rbcovprop} presents a simple coupled system of first order ordinary vector-matrix differential equations that can be numerically integrated. Using the data from this numerical solution we can numerically evaluate \cref{eq:rbencov}. The evolution predicted by \cref{eq:rbenmean,eq:rbencov} agrees closely with that obtained through Monte Carlo simulations. Simulation results are presented in \cref{fig:KEHistory,fig:statesHistory}. The growth in the simulated $3\sigma$ bounds is expected since there is uncertainty in the external torques as well as the initial angular velocities. The simulation parameters in SI units are as follows:
\begin{enumerate}
\item ${\bm Q} = \diag(0.005, 0.002, 0.003)~\mathrm{N^2 m^2}$
\item ${\bm J} = \diag(10,12,14)~\mathrm{kg m^2}$
\item Simulation time step for numerical integration : 0.1~s 
\item Total simulation time : 100 s
\item Initial mean and covariance of angular velocity assuming Gaussian distribution: \\ $[0.02,0.02,0.02]~\mathrm{rad/s}$ and $\diag(0.00002,0.00002,0.00002)~\mathrm{rad^2/s^2}$ respectively
\item Number of Monte Carlo sample points : $10000$
\end{enumerate}
 
%In the legend, A refers to plots obtained by integration of analytically found solution and MC refers to plots obtained by Monte Carlo simulation. We can see that there is a close match between the analytical and Monte Carlo plots. 
  
\begin{figure}[h!]
\centering
    \includegraphics[width=0.8\textwidth]{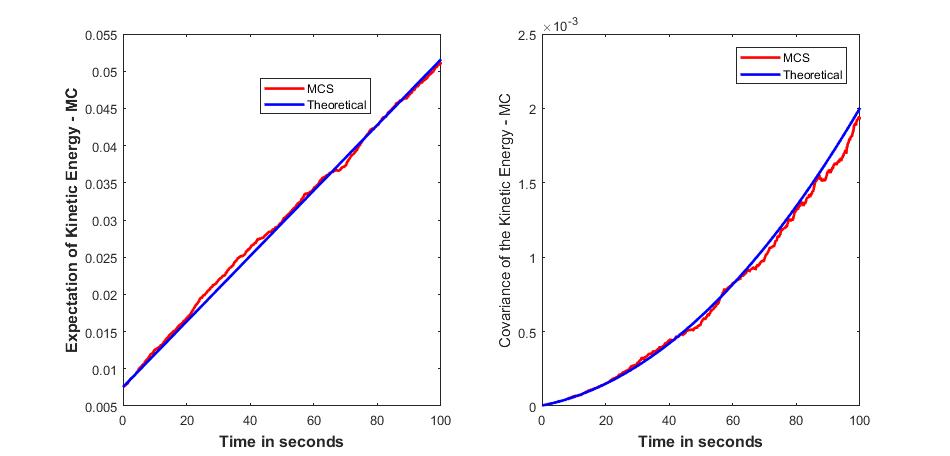}
    \caption{Mean and Covariance of Kinetic energy}
    \label{fig:KEHistory}
\end{figure}

\begin{figure}[h!]
\centering
    \includegraphics[width=0.8\textwidth]{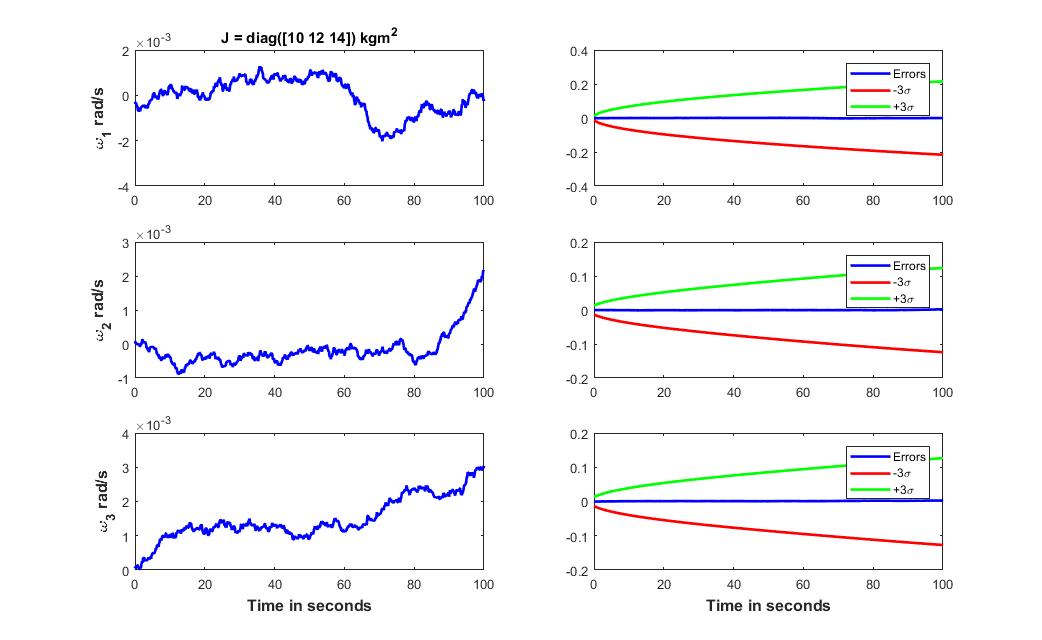}
    \caption{Angular Velocity Histories}
    \label{fig:statesHistory}
\end{figure}

\section{Analysis of Two Body Problem}

For the two body problem, the dynamics of the state (relative position $\r \in \R^3$ and relative velocity $\dot{\r} \in \R^3$)  is governed by \cref{tbd} where the $\mu$ is the graviational constant and perturbation accelerations are $\bm{\eta}\in \R^3$. 
\eqn{
\frac{d}{dt} \begin{bmatrix} \r \\ \dot{\r} \end{bmatrix} = \begin{bmatrix} \dot{\r} \\ -\frac{\oldmu}{r^3} \r \end{bmatrix} + \begin{bmatrix} \bm{0}_{3 \times 3} \\ \bm{I}_{3 \times 3} \end{bmatrix} \bm{\eta}
\label{tbd}
}
Clearly, as with the rigid body, we identify $\F(t,(\r,\dot{\r})) = \begin{bmatrix} \dot{\r} \\ -\frac{\oldmu}{r^3} \r \end{bmatrix} $ and $\G(t,(\r,\dot{\r}))=\begin{bmatrix} \bm{0}_{3 \times 3} \\ \bm{I}_{3 \times 3} \end{bmatrix}$. 
For the two body problem without perturbation, a few of the invariants are the semi-major axis of the orbit, angle of inclination, orbit eccentricity, longitude of right ascension, the argument of periapsis, and the time since periapsis passage. In this paper, we will study the specific angular momentum, and the total specific mechanical energy. \textbf{Note}, the total specific mechanical energy, the specific angular momentum and the orbit eccentricity are all related so only two out of these three can be treated as being independent.
The invariants are functions of $\r$ and $\dot{\r}$. We will look at evolution of the first and second moments of the square of the magnitude of angular momentum in the presence of stochastic torques ($\Q$ will denote the covariance of $\bm{\eta}$). 

\subsection{Specific Angular Momentum}

The specific angular momentum is defined as\cite{curtis}
\eqn{
    \bm{h}(\r,\dot{\r}) := \r \times \dot{\r}
}

Angular momentum is a vector invariant for the system in \cref{tbd} \cite{curtis}. The square of the Euclidean norm of the specific angular momentum ($h$) is considered as a scalar invariant for the system.
\eqn{
    \begin{aligned}
    {h}(\r,\dot{\r}) :=& ||\bm{h}||^2 \\
    =& (\r \times \dot{\r}) \cdot (\r \times \dot{\r})\\
    =& ||\r||^2 ||\dot{\r}||^2 - (\r \cdot \dot{\r})^2
    \end{aligned}
}

Simple computations yield:
\eqn{
    \frac{\partial h}{\partial \r} = 2||\dot{\r}||^2\r - 2(\dot{\r}\cdot\r)\dot{\r}
}
\eqn{
    \frac{\partial h}{\partial \dot{\r}} = 2||{\r}||^2\dot{\r} - 2(\dot{\r}\cdot\r){\r}
}

\eqn{
   \bm{H}_h = \begin{bmatrix} ||\dot{\r}||^2\bm{I}_3 -\dot{\r}\dot{\r}^T & 2\r \dot{\r}^T - \dot{\r}{\r}^T - (\r \cdot \dot{\r})\bm{I}_3 \\ 2\r \dot{\r}^T - \dot{\r}{\r}^T - (\r \cdot \dot{\r})\bm{I}_3 & ||\r||^2\bm{I}_3 - \r \r^T \end{bmatrix}
}

\eqn{
\G^T\bm{H}_h\G\Q = ||\r||^2\Q - \r \r^T\Q
\label{eq:htbdhess}
}

\noindent \underline{Evolution of Mean:}

Given
\begin{equation}
	\oldmu_h(t) := E[h(\r,\dot{\r})](t)
%\label{expen1}
\end{equation}
and defining $\lambda_m$ and $\lambda_M$ to be the minimum and maximum eigenvalue of $\Q$ respectively. 
Then,
\eqn{
E[||\r||^2](\tr(Q) - \lM(\Q)) \le \dot{\oldmu}_h \le E[||\r||^2](\tr(Q) - \lm)
}

The derivation of these inequalities is as follows. Using eq. \eqref{eq:meaninv} and \eqref{eq:htbdhess} 
\eqn{
\dot{\oldmu}_h(t) = E[||\r||^2\tr(\Q) - \r^T\Q\r]
}
which can be expressed in terms of second moments of $\r$ using Corollary \ref{eq:corrprop}. 
Since $\Q$ is symmetric, $\lm||\r||^2 \le \r^T\Q\r \le \lM||\r||^2$.
\eqn{
||\r||^2\tr(\Q) - \r^T\Q\r \le ||\r||^2\tr(\Q) - ||\r||^2\lm \ \forall \r
}
\eqn{
\begin{aligned}
    E[||\r||^2\tr(\Q) - \r^T\Q\r] =& \int (||\r||^2\tr(\Q) - \r^T\Q\r) p(\xi) d\xi \\
    \le &  \int (||\r||^2\tr(\Q) - ||\r||^2\lm p(\xi) d\xi \\
    =& E[||\r||^2](\tr(Q) - \lm)
\end{aligned}
}
The other inequality can be obtained similarly.

This gives us analytical bounds on the stochastic quantity $\oldmu_h$ as we can obtain $E[||\r||^2]$ from the equations for the second moment of the state (Corollary \ref{eq:corrprop}).

\noindent \underline{Evolution of Correlation:}

\eqn{
R_h(t) := E[h^2(\r,\dot{\r})](t)
}

Using eq. \eqref{eq:corrinv}
\eqn{
\begin{aligned}
\dot{R}_h(t) =& E[2(||\r||^2\tr(\Q) - \r^T\Q\r)(||\r||^2 ||\dot{\r}||^2 - (\r \cdot \dot{\r})^2) + 4||\r||^4\dot{\r}^T\Q\dot{\r} \\  &+ 4(\r \cdot \dot{\r})^2\r^T\Q\r -8||\r||^2(\r \cdot \dot{\r})\dot{\r}^T\Q\r]
\end{aligned}
}
which contains higher moments that can be expressed in terms of first and second moments using moment generating function of Gaussian random variable. We will now get analytical bounds on $\dot{R}_h(t)$ also. We will define $\v$ to make this derivation easier since it will recur throughout the derivation. 
\eqn{
\v := ||\r||^2\dot{\r} - (\r \cdot \dot{\r})\r = ||\r||^2(\dot{\r} - (\frac{\r}{||\r||} \cdot \dot{\r})\frac{\r}{||\r||})
}

Notice that $\v$ is just $\dot{\r}$ with the component of $\r$ removed from it. Hence $\v$ can never be parallel to $\r$. If $\dot{\r}$ is parallel to $\r$ then $\v = \bm{0}$. For convenience, define $\{\b_1,\b_2\}$ as unit norm vectors orthogonal to $\r$ such that $\{\b_1,\b_2,\r\}$ forms a basis for $\R^3$. Thus $\v$ can take all values in $\text{span}\{{\b_1,\b_2}\}$ but not in entire $\R^3$. The following computations will cast $\dot{R}_h$ in terms of $\v$ and $\r$. 

\eqn{
||\v||^2 = ||\r||^4||\dot{\r}||^2 - ||\r||^2(\r \cdot \dot{\r})^2
}

\eqn{
||\r||^4\dot{\r}^T\Q\dot{\r} + (\r \cdot \dot{\r})^2\r^T\Q\r -2||\r||^2(\r \cdot \dot{\r})\dot{\r}^T\Q\r = \v^TQ\v
}

\eqn{
\dot{R}_h = 2E[\tr(\Q)||\v||^2 - \frac{\r^T \Q \r}{||\r||^2}||\v||^2 + 2\v^T \Q \v]
\label{eq:Rhdot}
}

Let us consider two cases: (i) $\lm= \lM \iff \Q = p\bm{I}_{3 \times 3}$ for some $p \in \mathbb{R}$ (ii) $\lm \neq \lM$.

For case (i), 

\eqn{
\dot{R}_h = 24pE[||\v||^2]
\label{eq:Rhdot1}
}

For case (ii), 

Let $\dot{\r} = \theta \r + \theta_1\b_1 + \theta_2\b_2$ for some $\theta,\theta_1,\theta_2 \in \R$

From definition, $\v = ||\r||^2(\theta\r + \theta_1\b_1 + \theta_2\b_2) - \theta(\r \cdot \r)\r = ||\r||^2(\theta_1 \b_1 + \theta_2 \b_2)$. Hence $\v$ takes all values in the $\text{span}\{{\b_1,\b_2}\}$ but not in entire $\R^3$. 

Let us focus on getting the upper bound on $\dot{R}_h$. The procedure for the upper bound will be analogous. We can choose $\r$ such that $\r$ is the eigenvector corresponding to minimum eigenvalue of $\Q$ and $\v$ to be eigenvector corresponding to maximum eigenvalue of $\Q$. %$\r$ can be chosen without trouble $\r$ can take any value in $\R^3$, in particular, the eigenvector in question. 
Now since eigenvectors corresponding to distinct eigenvalues are orthogonal for a symmetric matrix (and the minimum and maximum eigenvalues are distinct from assumption) the eigenvector corresponding to the maximum eigenvalue of $\Q$ will be orthogonal to chosen $\r$ hence $\v$ can take that value. 

\noindent \textit{Outline of the Proof:}
For all $\r,\v$:
\eqn{
\begin{aligned}
\tr(\Q)||\v||^2 - \frac{\r^TQ\r}{||\r||^2}||\v||^2 + 2\v^TQ\v \le & \tr(\Q)||\v||^2 - \lm||\v||^2 + 2\v^T \Q \v \\
\le & \tr(\Q)||\v||^2 - \lm||\v||^2 + 2||\v||^2\lM \qed
\end{aligned}
}

Thus we have,
\eqn{
2(\tr(\Q) + 2\lm - \lM) E[||\v||^2] \le \dot{R}_h \le 2(\tr(\Q) + 2\lM - \lm) E[||\v||^2]
\label{eq:Rhdot2}
}

Notice that $||\v||^2 = h||\r||^2$. This simplifies \cref{eq:Rhdot} to
\eqn{
\dot{R}_h = 2E[h\tr(\Q)||\r||^2 - h\r^T \Q \r + 2\v^T \Q \v]
}

\cref{eq:Rhdot1} to
\eqn{
\dot{R}_h = 24pE[h||\r||^2]
}

\cref{eq:Rhdot2} to
\eqn{
2(\tr(\Q) + 2\lm - \lM) E[h||\r||^2] \le \dot{R}_h \le 2(\tr(\Q) + 2\lM - \lm) E[h||\r||^2]
}

\noindent \underline{Summary:}

The bounds obtained can be summarized as follows: 
\begin{equation*}
    E[||\r||^2](\tr(Q) - \lM) \le \dot{\oldmu}_h \le E[||\r||^2](\tr(Q) - \lm)
\end{equation*}
If $\Q$ is a multiple of identity ($p\bm{I}_{3 \times 3}$), then 
\begin{equation*}
    \dot{R}_h = 24pE[h||\r||^2]
\end{equation*}
and otherwise
\begin{equation*}
    2(\tr(\Q) + 2\lm - \lM)E[h||\r||^2] \le \dot{R}_h \le 2(\tr(\Q) + 2\lM - \lm)E[h||\r||^2]
\end{equation*}

\section{Conclusion}
In this paper, we considered dynamical systems with invariants when perturbed by Gaussian white noise. We first derived how the expectation of any function of the state of the perturbed system evolves with time. We used this to study the temporal evolution of the first two statistical moments of the system's invariants. Two case studies were investigated, first the kinetic energy of a rigid body and the second the square of the norm of the specific angular momentum in the two body problem. In the rigid body case, the propagation of the mean of kinetic energy has a linear evolution with time and covariance has a numerically implementable structure. Numerical simulations were performed and the semi-analytical solutions were compared with Monte Carlo simulations. For the two body problem, bounds were established for the mean and covariance of the angular momentum.
\appendix

\section{Appendix A}

%moment generating function of multivariate Gaussian to express higher moments in terms of $\Sigma,\mu$ and the expansion $\dot{\Sigma} = \dot{\bm{R}} - \mu\dot{\mu}^T - \dot{\mu}\mu^T$ we arrive at expressions for $\dot{\mu}, \dot{\Sigma}$. 

We will use corollary \ref{eq:meanprop} and \ref{eq:corrprop} to derive expressions for $\dot{\mu}$ and $\dot{\Sigma}$ in \cref{eq:rbmeanprop,eq:rbcovprop} respectively. We will do it here for the case when $\J$ is diagonal, as in \cref{rbd1,rbd2,rbd3}. From corollary \ref{eq:meanprop} and \ref{eq:corrprop},
\begin{align}
\dot{\mu} =& E[-\bm{\omega} \times \J \bm{\omega}] \\
\dot{\bm{R}} = \dot{\Sigma} =&  E[-\bm{\omega} (\J^{-1}(\bm{\omega} \times \J \bm{\omega}))^T] + E[-(\J^{-1}(\bm{\omega} \times \J \bm{\omega}))\bm{\omega}^T] + \J^{-1}\Q\J^{-1} \label{eq:Rdot}
\end{align}

Note for convenience that

\begin{align}
    -\J^{-1}(\bm{\omega} \times \J \bm{\omega}) =& \begin{bmatrix} c_1\omega_2\omega_3 \\ c_2 \omega_1\omega_3 \\ c_3\omega_1\omega_2 \end{bmatrix} \\ 
   -\bm{\omega} (\J^{-1}(\bm{\omega} \times \J \bm{\omega}))^T =& \begin{bmatrix} c_1\omega_1\omega_2\omega_3 & c_2\omega_1^2\omega_2 & c_3\omega_1^2\omega_2 \\ c_1\omega_2^2\omega_3 & c_2\omega_1\omega_2\omega_3 & c_3\omega_1\omega_2^2 \\ c_1\omega_2\omega_3^2 & c_2\omega_1\omega_3^2 & c_3\omega_1\omega_2\omega_3 \end{bmatrix}
\end{align}

Define 
\begin{align}
\Omega_1 :=& E[-\bm{\omega} \times \J \bm{\omega}] \\
\Omega_2 :=& E[-\bm{\omega} (\J^{-1}(\bm{\omega} \times \J \bm{\omega}))^T] \Rightarrow E[-(\J^{-1}(\bm{\omega} \times \J \bm{\omega}))\bm{\omega}^T] = \Omega_2^T
\end{align}

Using \cref{eq:covcorr} we establish $E[\omega_i\omega_j] = \oldSigma_{ij} + \oldmu_i\oldmu_j$. Thus we obtain $\Omega_1 = -\J^{-1}\left( \mu \times \J \mu \right) + \J^{-1}{\bm{\sigma}}$ with $\displaystyle \bm{\sigma} := \left[(J_2 - J_3)\oldSigma_{23}~\quad (J_3 - J_1)\oldSigma_{31}~\quad (J_1-J_2)\oldSigma_{12}\right]^T$ effectively arriving at \cref{eq:rbmeanprop}. To calculate $\Omega_2$, the expectation will thus involve third moments of the random variable $\bm{\omega}$.  We will use the moment generating function of the multivariate Gaussian to write the third moments in terms of first and second moments. For example, $E[\omega_1\omega_2\omega_3] = \oldmu_1\oldmu_2\oldmu_3 + \oldmu_1\oldSigma_{23} + \oldmu_2\oldSigma_{13} + \oldmu_3\oldSigma_{12}$ which is obtained by finding $\dfrac{\partial \varphi(\bm{t})}{\partial t_1 \partial t_2 \partial t_3}  \Big{|}_{\bm{t = 0}}$ where $\bm{t} = [t_1,t_2,t_3]$. Hence we have

%{\small
%\eqn{
%    \Omega_2 = \begin{bmatrix} c_1(\oldmu_1\oldmu_2\oldmu_3 + \oldmu_1\oldSigma_{23} + \oldmu_2\oldSigma_{13} + \oldmu_3\oldSigma_{12}) & c_2(\oldmu_1^2\oldmu_3 + 2\oldmu_1\oldSigma_{13} + \oldmu_3\oldSigma_{11})  & c_3(\oldmu_1^2\oldmu_2 + 2\oldmu_1\oldSigma_{12} + \oldmu_3\oldSigma_{11}) \\ c_1(\oldmu_2^2\oldmu_3 + 2\oldmu_2\oldSigma_{23} + \oldmu_3\oldSigma_{22}) &  c_2(\oldmu_1\oldmu_2\oldmu_3 + \oldmu_1\oldSigma_{23} + \oldmu_2\oldSigma_{13} + \oldmu_3\oldSigma_{12}) & c_3(\oldmu_1\oldmu_2^2 + 2\oldmu_2\oldSigma_{12} + \oldmu_1\oldSigma_{22}) \\ c_1(\oldmu_2\oldmu_3^2 + 2\oldmu_3\oldSigma_{23} + \oldmu_2\oldSigma_{33}) & c_2(\oldmu_1\oldmu_3^2 + 2\oldmu_3\oldSigma_{13} + \oldmu_1\oldSigma_{33}) & c_2(\oldmu_1\oldmu_2\oldmu_3 + \oldmu_1\oldSigma_{23} + \oldmu_2\oldSigma_{13} + \oldmu_3\oldSigma_{12})\end{bmatrix}
%}
%}

{
\eqn{
    \Omega_2 =\Omega_3 \begin{bmatrix} c_1 & 0 & 0 \\ 0 & c_2 & 0 \\ 0 & 0 & c_3\end{bmatrix}
}
}

with $C:= (\oldmu_1\oldmu_2\oldmu_3 + \oldmu_1\oldSigma_{23} + \oldmu_2\oldSigma_{13} + \oldmu_3\oldSigma_{12})$ and

{\small
\eqn{
    {\Omega}_3 := \begin{bmatrix} C & (\oldmu_1^2\oldmu_3 + 2\oldmu_1\oldSigma_{13} + \oldmu_3\oldSigma_{11})  & (\oldmu_1^2\oldmu_2 + 2\oldmu_1\oldSigma_{12} + \oldmu_3\oldSigma_{11}) \\ 
    (\oldmu_2^2\oldmu_3 + 2\oldmu_2\oldSigma_{23} + \oldmu_3\oldSigma_{22}) &  C & (\oldmu_1\oldmu_2^2 + 2\oldmu_2\oldSigma_{12} + \oldmu_1\oldSigma_{22}) \\
    (\oldmu_2\oldmu_3^2 + 2\oldmu_3\oldSigma_{23} + \oldmu_2\oldSigma_{33}) & (\oldmu_1\oldmu_3^2 + 2\oldmu_3\oldSigma_{13} + \oldmu_1\oldSigma_{33}) & C\end{bmatrix}
}
}

%{\small
%\eqn{
%    \Omega_2 = \begin{bmatrix} Cc_1 & c_2(\oldmu_1^2\oldmu_3 + 2\oldmu_1\oldSigma_{13} + \oldmu_3\oldSigma_{11})  & c_3(\oldmu_1^2\oldmu_2 + 2\oldmu_1\oldSigma_{12} + \oldmu_3\oldSigma_{11}) \\ 
%    c_1(\oldmu_2^2\oldmu_3 + 2\oldmu_2\oldSigma_{23} + \oldmu_3\oldSigma_{22}) &  c_2C & c_3(\oldmu_1\oldmu_2^2 + 2\oldmu_2\oldSigma_{12} + \oldmu_1\oldSigma_{22}) \\
%    c_1(\oldmu_2\oldmu_3^2 + 2\oldmu_3\oldSigma_{23} + \oldmu_2\oldSigma_{33}) & c_2(\oldmu_1\oldmu_3^2 + 2\oldmu_3\oldSigma_{13} + \oldmu_1\oldSigma_{33}) & c_3C\end{bmatrix}
%}
%}

Substituting in \cref{eq:Rdot} we obtain
\eqn{
\dot{\bm{R}} = \Omega_2 + \Omega_2^T + \J^{-1}\Q\J^{-1}
\label{eq:Rdot2}
}

We will now use \cref{eq:covcorrder} to arrive at $\dot{\Sigma}$. Note for convenience that
\eqn{
\mu\dot{\mu}^T = \begin{bmatrix} c_1(\oldmu_1\oldSigma_{23} + \oldmu_1\oldmu_2\oldmu_3) & c_2(\oldmu_1\oldSigma_{13} + \oldmu_3\oldmu_1^2) & c_3(\oldmu_1\oldSigma_{12} + \oldmu_1^2\oldmu_2) \\ c_1(\oldmu_2\oldSigma_{23} + \oldmu_2^2\oldmu_3) & c_2(\oldmu_2\oldSigma_{13} + \oldmu_1\oldmu_2\oldmu_3) & c_3(\oldmu_2\oldSigma_{12} + \oldmu_1\oldmu_2^2) \\ c_1(\oldmu_3\oldSigma_{23} + \oldmu_2\oldmu_3^2) & c_2(\oldmu_3\oldSigma_{13} + \oldmu_1\oldmu_3^2) & c_3(\oldmu_3\oldSigma_{12} + \oldmu_1\oldmu_2\oldmu_3) \end{bmatrix}
\label{eq:mumudot}
}

Define $\displaystyle \bm{A} := \frac{\partial}{\partial \alpha}\f(t,\alpha)\bigg\rvert_{\alpha = \mu}$. When calculated explicitly it evaluates to
\eqn{
\bm{A} = \begin{bmatrix} 0 & c_1\oldmu_3 & c_1\oldmu_2 \\ c_2\oldmu_3 & 0 & c_2\oldmu_1 \\ c_3\oldmu_2 & c_3\oldmu_1 & 0 \end{bmatrix}
\label{eq:A}
}

It remains to be noticed that substituting \cref{eq:mumudot,eq:Rdot2} in \cref{eq:covcorrder} and using \cref{eq:A} yields \cref{eq:rbcovprop}. 

\section{APPENDIX B}

Note for convenience that $\G = \J^{-1}$, $\HU = \frac{\J}{2}$, $U = \frac{1}{2}\bm{\omega}^T\J\bm{\omega}$, $\Us = \J\bm{\omega}$. We make these substitutions in eq. \eqref{eq:meaninv}, \eqref{covinv}. 

\eqn{
\begin{aligned}
\dot{\mu}_K =& \int  (\frac{\J}{2}\J^{-1}\Q\J^{-1}) p_{\S}d\zeta \\
=& \frac{1}{2}\tr(QJ^{-1})
\end{aligned}
}

\eqn{
\begin{aligned}
\dot{\Sigma}_K =& \int(2(\frac{1}{2}\bm{\omega}^T\J\bm{\omega} - \oldmu_K)\tr(\Q\J^{-1}\frac{\J}{2}\J^{-1}) + \tr(\bm{\omega}^T\J^T\J^{-1}\Q\J^{-1}\J\bm{\omega}))p_{\S} d\zeta \\
=& \oldmu_K - \oldmu_K + \int(\tr(\bm{\omega}^T\Q\bm{\omega}))p_{\S} d\zeta) \\
=& \tr(\Q\corr(\bm{\omega})) \\
=& \tr(\Q(\Sigma + \mu\mu^T))
\end{aligned}
}
%-----------------
\bibliographystyle{AAS_publication}   % Number the references.
\bibliography{references}   % Use references.bib to resolve the labels.

\begin{thebibliography}{10}

\bibitem{fp-astro}
E.~Kim, I.~Yoon, H.~M. Lee, and R.~Spurzem, ``Comparative study between N-body
  and Fokker–Planck simulations for rotating star clusters – I. Equal-mass
  system,''  {\em Monthly Notices of the Royal Astronomical Society}, Vol.~383,
  No.~1, 2008, pp.~2--10, 10.1111/j.1365-2966.2007.12524.x.

\bibitem{fp-physics}
C.-Z. Ning and G.~Hu, ``Exact Stationary Solution of Fokker-Planck Equation and
  Generalized Potential for Non-Equilibrium Systems Without Detailed Balance,''
   {\em Communications in Theoretical Physics}, Vol.~16, No.~4, 1991, p.~415.

\bibitem{fp-chem1}
R.~E.~D. McClung, ``The Fokker–Planck–Langevin model for rotational
  Brownian motion. I. General theory,''  {\em The Journal of Chemical Physics},
  Vol.~73, No.~5, 1980, pp.~2435--2442, 10.1063/1.440394.

\bibitem{fp-chem2}
G.~Levi, J.~P. Marsault, F.~Marsault-Herail, and R.~E.~D. McClung, ``The
  Fokker–Planck–Langevin model for rotational Brownian motion. II.
  Comparison with the extended rotational diffusion model and with observed
  infrared and Raman band shapes of linear and spherical molecules in fluids,''
   {\em The Journal of Chemical Physics}, Vol.~73, No.~5, 1980, pp.~2443--2453,
  10.1063/1.440395.

\bibitem{fp-chem3}
R.~E.~D. McClung, ``The Fokker–Planck–Langevin model for rotational
  Brownian motion. III. Symmetric top molecules,''  {\em The Journal of
  Chemical Physics}, Vol.~75, No.~11, 1981, pp.~5503--5513, 10.1063/1.441954.

\bibitem{risken}
H.~Risken and T.~Frank, {\em The Fokker-Planck Equation}.
\newblock Springer-Verlag Berlin Heidelberg, 1996.

\bibitem{junkins-optest}
J.~L. Crassidis and J.~L. Junkins, {\em Optimal Estimation of Dynamic Systems}.
\newblock Chapman \& Hall/CRC, 2nd~ed., 2011.

\bibitem{fuller-fpham}
A.~T. Fuller, ``Analysis of nonlinear stochastic systems by means of the
  Fokker-Planck equation,''  {\em International Journal of Control}, Vol.~9,
  No.~6, 1969, pp.~603--655.

\bibitem{singla-gmm}
G.~Terejanu, P.~Singla, T.~Singh, and P.~D. Scott, ``Uncertainty Propagation
  for Nonlinear Dynamic Systems Using Gaussian Mixture Models,''  {\em Journal
  of Guidance, Control, and Dynamics}, Vol.~31, Nov 2008, pp.~1623--1633,
  10.2514/1.36247.

\bibitem{singla-gsf}
G.~Terejanu, P.~Singla, T.~Singh, and P.~D. Scott, ``Adaptive Gaussian Sum
  Filter for Nonlinear Bayesian Estimation,''  {\em IEEE Transactions on
  Automatic Control}, Vol.~56, Sept 2011, pp.~2151--2156,
  10.1109/TAC.2011.2141550.

\bibitem{sorenson-1}
H.~Sorenson and D.~Alspach, ``Recursive bayesian estimation using gaussian
  sums,''  {\em Automatica}, Vol.~7, No.~4, 1971, pp.~465 -- 479,
  https://doi.org/10.1016/0005-1098(71)90097-5.

\bibitem{sorenson-2}
D.~Alspach and H.~Sorenson, ``Nonlinear Bayesian estimation using Gaussian sum
  approximations,''  {\em IEEE Transactions on Automatic Control}, Vol.~17,
  August 1972, pp.~439--448, 10.1109/TAC.1972.1100034.

\bibitem{singla-split}
K.~Vishwajeet and P.~Singla, ``Adaptive splitting technique for Gaussian
  mixture models to solve Kolmogorov Equation,''  {\em 2014 American Control
  Conference}, June 2014, pp.~5186--5191, 10.1109/ACC.2014.6859240.

\bibitem{park-nlg}
R.~S. Park and D.~J. Scheeres, ``Nonlinear Mapping of Gaussian Statistics:
  Theory and Applications to Spacecraft Trajectory Design,''  {\em Journal of
  Guidance, Control, and Dynamics}, Vol.~29, Nov 2006, pp.~1367--1375,
  10.2514/1.20177.

\bibitem{lee-1}
T.~Lee, M.~Leok, and N.~H. McClamroch, ``Global symplectic uncertainty
  propagation on SO(3),''  {\em 2008 47th IEEE Conference on Decision and
  Control}, Dec 2008, pp.~61--66, 10.1109/CDC.2008.4739058.

\bibitem{kdm-1}
K.~J. DeMars, R.~H. Bishop, and M.~K. Jah, ``Entropy-Based Approach for
  Uncertainty Propagation of Nonlinear Dynamical Systems,''  {\em Journal of
  Guidance, Control, and Dynamics}, Vol.~36, May 2013, pp.~1047--1057,
  10.2514/1.58987.

\bibitem{kdm-2}
J.~Darling and K.~DeMars, ``Uncertainty Propagation of Correlated Quaternion
  and Euclidean States Using the Gauss-Bingham Density,''  {\em Journal of
  Advances in Information Fusion}, Vol.~11, 12 2016, pp.~186--205.

\bibitem{lee-2}
T.~Lee, ``Stochastic optimal motion planning and estimation for the attitude
  kinematics on SO(3),''  {\em 52nd IEEE Conference on Decision and Control},
  Dec 2013, pp.~588--593, 10.1109/CDC.2013.6759945.

\bibitem{sanyal-1}
A.~K. Sanyal, T.~Lee, M.~Leok, and N.~H. McClamroch, ``Global optimal attitude
  estimation using uncertainty ellipsoids,''  {\em Systems \& Control Letters},
  Vol.~57, No.~3, 2008, pp.~236--245.

\bibitem{wiesel-polhode}
W.~E. Wiesel, {\em Spaceflight Dynamics}.
\newblock Mc-Graw Hill, 2nd~ed., 1997.

\bibitem{hughes-polhode}
P.~C. Hughes, {\em Spacecraft Attitude Dynamics}.
\newblock Dover Publications, 2004.

\bibitem{evans-bm}
L.~Evans, {\em An Introduction to Stochastic Differential Equations}.
\newblock American Mathematical Society, 2013.

\bibitem{GRV}
D.~S. Tracy and S.~Sultan, ``Higher order moments of multivariate normal
  distribution using matrix derivatives,''  {\em Stochastic Analysis and
  Applications}, Vol.~11, 1993, pp.~337--348.

\bibitem{sflow1}
A.~A. Dorogovtsev and I.~I. Nishchenko, ``An analysis of stochastic flows,''
  {\em Stochastic Analysis and Applications}, Vol.~8, 2014, pp.~331--342.

\bibitem{sflow2}
H.~Kunita, {\em Lectures on Stochastic Flows And Applications}.
\newblock Springer-Verlag, 1986.
\newblock http://www.math.tifr.res.in/~publ/ln/tifr78.pdf.

\bibitem{singla-qgkf}
S.~Chakravorty, M.~Kumar, and P.~Singla, ``A quasi-Gaussian Kalman filter,''
  {\em 2006 American Control Conference}, June 2006, pp.~6 pp.--,
  10.1109/ACC.2006.1655484.

\end{thebibliography}
%-----------------
\end{document}